\documentclass[11pt,a4paper]{article}
\usepackage[margin=25mm]{geometry}
\usepackage{amsmath,amssymb,amsthm,bm}
\usepackage{graphicx}
\usepackage{booktabs,array,longtable}
\usepackage[hidelinks]{hyperref}
\usepackage[numbers,sort&compress]{natbib}
\usepackage{microtype}
\usepackage{caption}
\usepackage{xcolor}
\DeclareTextSymbol{\textsection}{OMS}{120}\DeclareTextSymbolDefault{\textsection}{OMS}
\DeclareTextSymbol{\textasteriskcentered}{OMS}{3}\DeclareTextSymbolDefault{\textasteriskcentered}{OMS}
\DeclareTextSymbol{\textbullet}{OMS}{15}\DeclareTextSymbolDefault{\textbullet}{OMS}
\newtheorem{theorem}{Theorem}
\newtheorem{lemma}[theorem]{Lemma}
\newtheorem{proposition}[theorem]{Proposition}

\theoremstyle{definition}
\newtheorem{definition}[theorem]{Definition}
\newtheorem{assumption}[theorem]{Assumption}
\newtheorem{facts}[theorem]{Facts}
\newtheorem{example}[theorem]{Example}
\theoremstyle{definition}
\newtheorem{remark}[theorem]{Remark}

\newcommand{\R}{\mathbb{R}}
\newcommand{\Rm}{\bm{R}}
\newcommand{\rv}{\bm{r}}
\newcommand{\Cm}{\Rm_{(X,Y)}}
\newcommand{\ah}{\hat{\bm b}}            
\newcommand{\ahs}{\hat{\bm b}_{\#}}      
\newcommand{\bsh}[1]{\hat b_{#1}^{\#}}    
\newcommand{\rhoM}{\rho_{X\vdash Y}}
\newcommand{\rhoP}[1]{\rho_{#1 Y\mid(-#1)}}
\newcommand{\nR}[1]{\|#1\|_{\Rm^{-1}}}

\newcommand{\tr}{^{\top}}
\DeclareMathOperator{\sd}{sd}
\DeclareMathOperator{\sgn}{sgn}

\DeclareMathOperator{\corr}{corr}
\DeclareMathOperator{\VIF}{VIF}
\DeclareMathOperator{\Var}{Var}
\DeclareMathOperator{\Cov}{Cov}
\newcommand{\versiontag}{(v29)}  
\newcommand{\lab}[1]{\leavevmode\par\noindent\textbf{\upshape(#1)}\ }
\newcommand{\labfirst}[1]{\textbf{\upshape(#1)}\ }

\title{Reading Multiple Regression from an Ellipsoid Inscribed in the Unit Hypercube}
\author{Toshiyuki Shimono\thanks{Research Organization of Information and Systems (ROIS), Tokyo, Japan.
E-mail: \texttt{shimono@ism.ac.jp}}}
\date{5 October 2026 \versiontag}

\begin{document}
\maketitle

\begin{abstract}
Consider the linear regression of a response $Y$ on explanatory variables $X_1,\dots,X_p$ with a positive definite joint correlation matrix. Let $\Rm$ be the correlation matrix of the explanatory variables and $\rv$ the vector of their correlations with $Y$.
The ellipsoid $E=\{\bm{x} : \bm{x}\tr\Rm^{-1}\bm{x}=1\}$ is inscribed in the cube $[-1,1]^p$ and touches the face $x_j=1$ at $T_j$, the $j$th column of $\Rm$.
The point $P=\rv$ lies inside $E$, and three kinds of regression quantities are read off this diagram by elementary geometric constructions.
The multiple correlation is $|OP|/|OP'|$, where $O$ is the origin and $P'$ is the point where the ray from $O$ through $P$ meets $E$ (for $P\ne O$).
The standardized regression coefficients are the coordinates of $P$ in the basis $T_1,\dots,T_p$.
The partial correlation between $X_i$ and $Y$ given the other explanatory variables is the position of $P$ on the chord of $E$ through $P$ parallel to the $x_i$-axis, on a linear scale from $-1$ to $1$ increasing with $x_i$; furthermore, the parallel chord through the origin has half-length $1/\sqrt{\VIF_i}$.
Examples illustrate sign reversal and standardized coefficients outside the interval $[-1,1]$.
\end{abstract}


\section{Introduction}\label{sec:intro}

The linear model
\[
\hat Y=\hat a+\hat b_1X_1+\dots+\hat b_pX_p
\]
is widely used, but when the explanatory variables are correlated it is hard to see from the numbers alone how its coefficients, the marginal correlations and the partial correlations are related.
A coefficient $\hat b_i$ may have the sign opposite to that of the marginal correlation between $X_i$ and $Y$ (and then so does the partial correlation, which has the sign of $\hat b_i$); a standardized coefficient may exceed $1$ in absolute value; and the multiple correlation may be much larger than the absolute value of every marginal correlation.
These phenomena can be explained algebraically under the headings of suppression, multicollinearity and enhancement, and have been analyzed, with accompanying graphs, by \citet{Hamilton1987}, \citet{Shieh2001}, \citet{FriedmanWall2005} and \citet{Waller2011}.
This paper develops one diagram, drawn from the correlations alone, from which three kinds of regression quantities---the multiple correlation, the $p$ standardized coefficients and the $p$ partial correlations---are read off by elementary geometric constructions.

Elliptical geometry has long served regression: the correlation coefficient itself is read from the standardized ellipse \citep[Sects.~3 and~4.1]{FriendlyMonetteFox2013}, and the ellipse of concentration \citep[Sect.~21.10]{Cramer1946}, \citep[Sect.~7.3]{Dempster1969}, called the data ellipse by \citet{FriendlyMonetteFox2013}, makes the ellipse a working tool for multivariate data.
For a \emph{multiple} regression, two common geometric approaches are vectors and projections in the space of observations \citep[Chap.~3]{TakeuchiYanai1972}, \citep{SavilleWood1991,Wickens1995,ONeill2019}, and data ellipsoids in the space of the variables, including three-dimensional displays of the regression plane \citep[Sect.~4.8, Fig.~17]{FriendlyMonetteFox2013}.

What this paper adds is a single configuration in $\R^p$. Its inputs are the $p\times p$ correlation matrix $\Rm$ of the explanatory variables and the vector $\rv$ of their correlations with $Y$. The ellipsoid $E=\{\bm{x} : \bm{x}\tr\Rm^{-1}\bm{x}=1\}$---the data ellipsoid of the standardized explanatory variables---is inscribed in the cube $[-1,1]^p$ and touches its faces at the columns of $\pm\Rm$ (Lemma~\ref{lem:inscribed}). The point $P=\rv$ is placed in the same coordinate space; it lies on or inside $E$, and strictly inside when the full correlation matrix of $(X,Y)$ is positive definite (Lemma~\ref{lem:inside}). Theorems~\ref{thm:multiple}, \ref{thm:coef} and~\ref{thm:partial} show how this configuration yields the multiple correlation, standardized regression coefficients and partial correlations. The construction is valid in any finite dimension $p$, while direct visualization is generally practical for $p\le3$.

Section~\ref{sec:consequences} turns the diagram into statements about sign reversal, suppression, bounds on coefficients and variance inflation.
Section~\ref{sec:example} gives two worked examples, and Section~\ref{sec:related} compares the construction with the related elliptical and vector-geometric methods and says which parts of it we have not found elsewhere.
Appendix~\ref{app:symbols} lists the main symbols used in the paper.

Preliminary versions of Theorems~\ref{thm:multiple}, \ref{thm:coef} and~\ref{thm:partial} were announced without proofs in a conference abstract \citep{Shimono2019}; the present paper gives complete proofs, the closed form of the ruler reading and the consequences.

\section{Notation and classical facts}\label{sec:notation}

Throughout, $X=(X_1,\dots,X_p)\tr$ are the explanatory variables and $Y$ the response.
We work in parallel with a population setting and a sample setting. All variables have finite positive variances in the population setting and positive sample variances in the sample setting (the sample has $n$ rows). Regression means best linear prediction with an intercept in the former and ordinary least squares with an intercept in the latter; equivalently, the slopes can be computed from centered variables or data. In the sample setting, the same divisor is used for all variances and covariances. The formulas \eqref{eq:normal} to \eqref{eq:partial-inverse} below have the same form in both settings.

\begin{definition}[correlation notation]\label{def:corr}
$\rho_{ij}:=\corr(X_i,X_j)$, $\rho_i:=\corr(X_i,Y)$, $\Rm:=(\rho_{ij})\in\R^{p\times p}$, $\rv:=(\rho_1,\dots,\rho_p)\tr\in\R^p$, and the \textbf{full correlation matrix} of $(X,Y)$ is denoted by
\[
\Cm:=\begin{pmatrix}\Rm&\rv\\ \rv\tr&1\end{pmatrix}\in\R^{(p+1)\times(p+1)}.
\]
We write $X_{(-i)}$ for the vector of explanatory variables with $X_i$ removed; $\Rm_{(-i)(-i)}$, $\Rm_{(-i)i}$, and $\rv_{(-i)}$ denote, respectively, $\Rm$ with row and column $i$ deleted, column $i$ of $\Rm$ with entry $i$ deleted, and $\rv$ with entry $i$ deleted.
\end{definition}

\begin{definition}[multiple and partial correlation notation]\label{def:mpcorr}
For a set $A$ of variables and a variable $B$, $\rho_{A\vdash B}\in[0,1]$ denotes the \textbf{multiple correlation coefficient} of the linear regression of $B$ on $A$, as defined by \eqref{eq:multiple} below.
$\rhoP{i}$ is the \textbf{partial correlation} between $X_i$ and $Y$ given $X_{(-i)}$: the correlation between the residuals of $X_i$ and of $Y$ from their linear regressions on $X_{(-i)}$, when both residual variances are positive \citep[p.~270, Eq.~(4g.2.9)]{Rao1973}.
(For $p=1$ the conditioning set is empty, the residuals are the centered versions of $X_1$ and $Y$, and $\rhoP{1}=\rho_1$; with the convention that empty sums are zero, \eqref{eq:partial-classical} below also gives $\rho_1$.)
\end{definition}

\begin{assumption}[positive definiteness]\label{ass:pd}
Throughout this article:
\begin{equation}\label{eq:pd}\tag{A}
\Rm\succ0,\qquad\text{i.e., $\Rm$ is positive definite.}
\end{equation}
\end{assumption}
Then $\Cm\succeq0$, i.e., $\Cm$ is positive semidefinite; moreover, $\Cm\succ0$ unless $Y=a+\sum_i b_iX_i$ for some constants $a,b_1,\dots,b_p$, almost surely in the population setting or as an exact relation among the data columns and the constant column in the sample setting.

We use three classical facts, (F1) to (F3) in Facts~\ref{facts:classical} below, concerning best linear prediction and correlations between linear regression residuals. No normality assumption is required for the identities used here \citep[p.~39]{Anderson2003}. \citet[Sect.~2.5, pp.~33--41]{Anderson2003} and \citet[Sects.~4g.1--4g.2, pp.~266--270]{Rao1973} give covariance-form treatments, while \citet[pp.~68--78]{TakeuchiYanai1972} give the correlation-matrix forms used below; equation-by-equation sources are listed in Appendix~\ref{app:sources}.

\begin{definition}[standardized regression coefficients]\label{def:coefficients}
Let $\ah=(\hat b_1,\dots,\hat b_p)\tr$ be the vector of regression coefficients in $\hat Y=\hat a+\ah\tr X$, where $\hat a$ is the intercept. Define the \textbf{standardized regression coefficients} by
\[
\bsh{i}:=\frac{\sd(X_i)}{\sd(Y)}\hat b_i\quad(1\le i\le p),\qquad \ahs:=(\bsh{1},\dots,\bsh{p})\tr .
\]
Equivalently, these are the regression coefficients after each variable is centered and divided by its standard deviation.
\end{definition}

\begin{facts}[classical facts]\label{facts:classical}\leavevmode
\begin{itemize}
\item[(F1)] \textbf{Standardized regression coefficients.} From the normal equations for the standardized variables,
\begin{equation}\label{eq:normal}
\ahs=\Rm^{-1}\rv.
\end{equation}
\item[(F2)] \textbf{Multiple correlation.}
\begin{equation}\label{eq:multiple}
\rhoM=\sqrt{\rv\tr\Rm^{-1}\rv}.
\end{equation}
Using \eqref{eq:normal}, this is also $\rhoM=\sqrt{\ahs\tr\rv}$, and the residual variance is
\begin{equation}\label{eq:residual-variance}
\Var(Y-\hat a-\ah\tr X)=\sd(Y)^2(1-\rhoM^2).
\end{equation}
The determinant formula for partitioned matrices on $\Cm$ gives
\begin{equation}\label{eq:det}
\det\Cm=\det\Rm\cdot(1-\rv\tr\Rm^{-1}\rv)=\det\Rm\cdot(1-\rhoM^2).
\end{equation}
\item[(F3)] \textbf{Partial correlation.} Suppose that $\Cm\succ0$.
\begin{equation}\label{eq:partial-classical}
\rhoP{i}=\frac{\rho_i-\rv_{(-i)}\tr\Rm_{(-i)(-i)}^{-1}\Rm_{(-i)i}}{\sqrt{1-\rho^2_{X_{(-i)}\vdash X_i}}\;\sqrt{1-\rho^2_{X_{(-i)}\vdash Y}}}
\end{equation}
and, equivalently,
\begin{equation}\label{eq:partial-inverse}
\rhoP{i}=-\frac{R^{\,i,p+1}_{(X,Y)}}{\sqrt{R^{\,i,i}_{(X,Y)}\,R^{\,p+1,p+1}_{(X,Y)}}}\quad\text{with }R^{\,i,j}_{(X,Y)}:=(\Cm^{-1})_{ij}.
\end{equation}
\end{itemize}
\end{facts}

Vectors in $\R^p$ are columns; $\bm{e}_1,\dots,\bm{e}_p$ denote the standard basis vectors ($\bm{e}_i$ has a $1$ in position $i$ and $0$ elsewhere), so that $\bm{e}_i\tr \bm{x}=x_i$ and $\bm M\bm{e}_i$ is the $i$-th column of a matrix $\bm M$; $\delta_{ij}$ is the Kronecker delta, and $|\cdot|$ the Euclidean length.

\paragraph{Geometric notation.}
For an ellipsoid $\{\bm{x} : \bm{x}\tr\bm H\bm{x}=1\}$ with $\bm H$ symmetric positive definite we use the inner product $\langle \bm{u},\bm{v}\rangle_{\bm H}=\bm{u}\tr\bm H\bm{v}$ and the norm $\|\bm{u}\|_{\bm H}=\langle \bm{u},\bm{u}\rangle_{\bm H}^{1/2}$.
Two directions $\bm{u},\bm{v}$ are \textbf{conjugate} relative to the ellipsoid if $\langle \bm{u},\bm{v}\rangle_{\bm H}=0$ (the terminology of Dempster \citep[Sect.~3.2]{Dempster1969}); the \textbf{diametral hyperplane} conjugate to a direction $\bm{d}$ is $\{\bm{x} : \langle \bm{d},\bm{x}\rangle_{\bm H}=0\}$, the hyperplane through the center that bisects every chord parallel to $\bm{d}$; the tangent hyperplane at a boundary point $\bm{x}_0$ is $\{\bm{x} : \langle \bm{x}_0,\bm{x}\rangle_{\bm H}=1\}$.

\section{The inscribed ellipsoid and the point $P$}\label{sec:geometry}

\begin{definition}[geometric configuration]\label{def:objects}
In $\R^p$ with orthogonal coordinates $x_1,\dots,x_p$ let
\begin{align*}
O&=(0,\dots,0)\tr\in\R^p &&\text{(the origin),}\\
U&=[-1,1]^p &&\text{(the unit hypercube, a square when $p=2$),}\\
E&=\{\bm{x}\in\R^p : \bm{x}\tr\Rm^{-1}\bm{x}=1\} &&\text{(the inscribed ellipsoid, a surface),}\\
\bar E&=\{\bm{x}\in\R^p : \bm{x}\tr\Rm^{-1}\bm{x}\le1\} &&\text{(the solid ellipsoid bounded by $E$),}\\
T_i&=\Rm \bm{e}_i=(\rho_{1i},\dots,\rho_{pi})\tr &&\text{(the $i$-th column of $\Rm$), $i=1,\dots,p$,}\\
P&=\rv=(\rho_1,\dots,\rho_p)\tr &&\text{(shown as an unfilled circle in Figures~\ref{fig:config} to~\ref{fig:wedge}),}\\
\ell_i&=\{P+t\bm{e}_i : t\in\R\} &&\text{(the line through $P$ parallel to the $x_i$-axis), $i=1,\dots,p$.}
\end{align*}
Here $\bm{x}\tr\Rm^{-1}\bm{x}=\nR{\bm{x}}^2$ (Section~\ref{sec:notation}), so $E=\{\bm{x} : \nR{\bm{x}}=1\}$ and $\bar E=\{\bm{x} : \nR{\bm{x}}\le1\}$.
Throughout, ``ellipsoid'' means the surface $E$ (on which tangency points, chord endpoints and the point $P'$ of Theorem~\ref{thm:multiple} lie); statements about volumes and \textbf{shadows} (orthogonal projections onto coordinate subspaces; the term follows Dempster \citep[Sect.~6.6]{Dempster1969}) refer to the solid $\bar E$.
The coordinates of $P$ are the correlations of $Y$ with $X$.
\end{definition}

The whole figure is determined by the $(p+1)\times(p+1)$ correlation matrix $\Cm$ and nothing else.
Figure~\ref{fig:config} shows it for $p=2$ (note that $\rho_{11}=\rho_{22}=1$ and $\rho_{12}=\rho_{21}\in[-1,1]$).

\begin{figure}[t]
\centering
\includegraphics[width=0.46\textwidth]{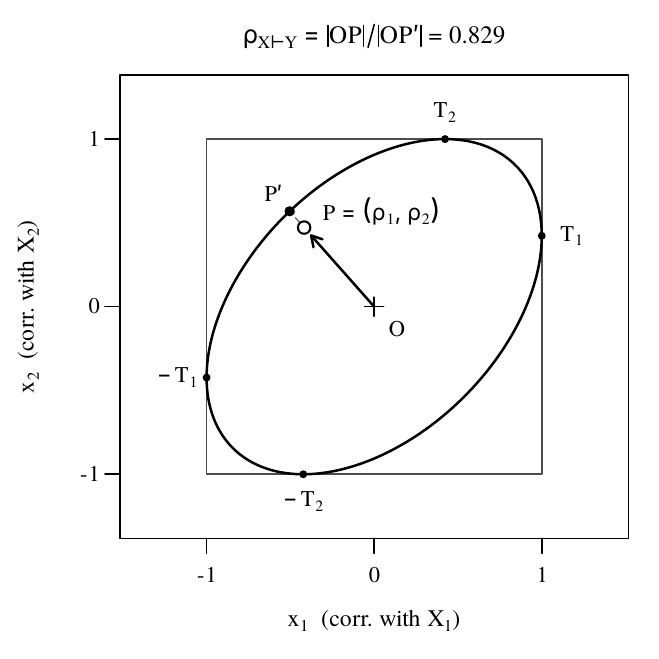}
\caption{The configuration for $p=2$ (Example~\ref{ex:baseball}: $\rho_{12}=0.423$, $\rv=(-0.418,0.472)\tr$).
The square $U$, the ellipse $E$ touching $U$ at $\pm T_1=\pm(1,\rho_{12})$ and $\pm T_2=\pm(\rho_{12},1)$, the point $P=\rv$, and the point $P'$ where the ray $OP$ meets $E$.
Theorem~\ref{thm:multiple}: the multiple correlation is $|OP|/|OP'|$. $O$ is marked by a cross, $P$ by an unfilled circle, and $P'$ and the tangency points by filled dots.}
\label{fig:config}
\end{figure}

\begin{lemma}[inscribed ellipsoid]\label{lem:inscribed}
Let $\bm M$ be a $p\times p$ matrix with $\bm M=\bm M\tr\succ0$ and unit diagonal, and let $E_{\bm M}=\{\bm{x} : \bm{x}\tr\bm M^{-1}\bm{x}=1\}$:
\lab{a} (tangency) $E_{\bm M}\subset U$, and $E_{\bm M}\cap\partial U=\{\pm\bm Me_1,\dots,\pm\bm Me_p\}$; the tangent hyperplanes of $E_{\bm M}$ at $\pm\bm M\bm{e}_i$ are $\{\bm{x} :x_i=\pm1\}$: they contain a pair of faces of $U$.
\lab{b} (uniqueness) If $\bm H=\bm H\tr\succ0$ and, for every $i$, the ellipsoid $\{\bm{x} : \bm{x}\tr\bm H\bm{x}=1\}$ is tangent to the face $x_i=1$ at the point $\bm M\bm{e}_i$, then $\bm H=\bm M^{-1}$.
\lab{c} (conjugacy) $\langle \bm M\bm{e}_i, \bm{e}_j\rangle_{\bm M^{-1}}=\delta_{ij}$: $\bm M\bm{e}_i$ is conjugate to $\bm{e}_j$, $j\ne i$; equivalently, the coordinate hyperplane $x_i=0$ is the diametral hyperplane conjugate to $\bm M\bm{e}_i$.
\end{lemma}

\begin{proof}
\labfirst{a} For $\bm{x}\in E_{\bm M}$ and each $i$, the Cauchy--Schwarz inequality in $\langle\cdot,\cdot\rangle_{\bm M^{-1}}$ gives
\[
|x_i|=|\bm{e}_i\tr \bm{x}|=|\langle \bm M\bm{e}_i,\bm{x}\rangle_{\bm M^{-1}}|\le\|\bm M\bm{e}_i\|_{\bm M^{-1}}\cdot\|\bm{x}\|_{\bm M^{-1}}=\sqrt{e_i\tr\bm M\bm{e}_i}\cdot1=1,
\]
with equality iff $\bm{x}=c\,\bm M\bm{e}_i$ for a scalar $c$; $\bm{x}\in E_{\bm M}$ forces $c^2 \bm{e}_i\tr\bm M\bm{e}_i=c^2=1$, i.e.\ $c=\pm1$.
The gradient of $\bm{x}\tr\bm M^{-1}\bm{x}$ at $\bm M\bm{e}_i$ is $2\bm M^{-1}\bm M\bm{e}_i=2\bm{e}_i$, so the tangent hyperplane there is $\{\bm{x} : \bm{e}_i\tr \bm{x}=\bm{e}_i\tr\bm M\bm{e}_i=1\}$.
\lab{b} Tangency to $x_i=1$ at $\bm M\bm{e}_i$ means that the gradient $2\bm H\bm M\bm{e}_i$ is parallel to $\bm{e}_i$, say $\bm H\bm M\bm{e}_i=\lambda_i\bm{e}_i$; since $\bm M\bm{e}_i$ lies on the ellipsoid, $1=(\bm M\bm{e}_i)\tr\bm H\bm M\bm{e}_i=\lambda_i\,\bm{e}_i\tr\bm M\bm{e}_i=\lambda_i$.
Hence $\bm H\bm M\bm{e}_i=\bm{e}_i$ for every $i$, so $\bm H\bm M=\bm I$.
\lab{c} $(\bm M\bm{e}_i)\tr\bm M^{-1}\bm{e}_j=\bm{e}_i\tr \bm{e}_j$.
\end{proof}

By Lemma~\ref{lem:inscribed}, $E$ touches the cube at $\pm T_i=\pm\Rm \bm{e}_i$ and is the unique origin-centered ellipsoid tangent to its faces at these prescribed points.
For $p=2$, $E$ is the ellipse through $\pm(1,\rho_{12})$ and $\pm(\rho_{12},1)$, with semi-axis $\sqrt{1+\rho_{12}}$ along $x_1=x_2$ and $\sqrt{1-\rho_{12}}$ along $x_1=-x_2$ (a circle when $\rho_{12}=0$); it is a level curve of the standardized bivariate normal density and corresponds, up to translation and uniform scaling, to the elliptical glyph of Murdoch and Chow \citep{MurdochChow1996} (a description of the shape only; normality is not assumed).

Conversely, let $\{\bm{x} : \bm{x}\tr\bm H\bm{x}=1\}$, $\bm H$ symmetric positive definite, be an origin-centered ellipsoid contained in $U$ and tangent to every face. The Cauchy--Schwarz argument of (a), with $\bm H^{-1}$ in place of $\bm M$, gives $\max\{|x_i| : \bm{x}\tr\bm H\bm{x}=1\}=\sqrt{(\bm H^{-1})_{ii}}$, attained at $\bm{x}=\pm\bm H^{-1}\bm{e}_i/\sqrt{(\bm H^{-1})_{ii}}$.
Tangency to the faces $x_i=\pm1$ means that this maximum is $1$ for every $i$; so $\bm M=\bm H^{-1}$ has unit diagonal, i.e.\ it is a positive definite correlation matrix, and the point of tangency with the face $x_i=1$ is $\bm M\bm{e}_i$, the $i$-th column of $\bm M$.
Hence $\bm M\mapsto E_{\bm M}$ is a bijection between positive definite correlation matrices and nondegenerate origin-centered ellipsoids tangent to every face of the cube, under which the $i$-th column of $\bm M$ is the point of tangency with the face $x_i=1$.

\begin{lemma}[$P$ lies on or inside $E$]\label{lem:inside}
\lab{a} $\nR{P}=\rhoM\le1$.
\lab{b} $\nR{P}=1$ iff $\Cm\nsucc0$, i.e.\ iff $Y-c$ is a linear combination of $X_1,\dots,X_p$ for some $c\in\R$.
\lab{c} $P$ is interior to $\bar E$ iff $\Cm\succ0$. The ellipsoid $E$ is nondegenerate by \eqref{eq:pd}.
\end{lemma}
\begin{proof}
By \eqref{eq:det}, $1-\nR{P}^2=\det\Cm/\det\Rm\ge0$ since $\Cm\succeq0$ and $\det\Rm>0$ by \eqref{eq:pd}.
\end{proof}

\section{The three theorems}\label{sec:theorems}

\subsection{Multiple correlation as a ratio of lengths}

\begin{theorem}[multiple correlation by a length ratio]\label{thm:multiple}
Let $P'$ be the intersection of $E$ with the ray from $O$ through $P$ (except when $P=O$):
\[
\frac{|OP|}{|OP'|}=\rhoM\qquad\text{(multiple correlation)}.
\]
\end{theorem}

\begin{proof}
Since $\nR{cP}=c\,\nR{P}$ for $c>0$, the point $P'=P/\nR{P}$ is the unique point of the ray with $\nR{\bm{x}}=1$.
Hence $|OP|/|OP'|=\nR{P}=\rhoM$ by Lemma~\ref{lem:inside}(a).
\end{proof}

For $p=2$, the ellipse $\rhoM E=\{\bm{x} : \nR{\bm{x}}=\rhoM\}$ passes through $P$, and the coefficient of determination $\rhoM^2$ is the ratio of the area it encloses to the area enclosed by $E$.

Theorem~\ref{thm:multiple} says that $\rhoM$ is the norm of $\rv$ in the metric $\Rm^{-1}$, i.e.\ a Mahalanobis length; in practice one never needs $\Rm^{-1}$: draw $E$ from the tangency points, mark $P$, and measure.
Three special positions are instructive.
(i) If $\Rm=\bm I$ then $E$ is the unit sphere and $\rhoM=|OP|=(\sum\rho_i^2)^{1/2}$.
(ii) If $P$ lies on the line through $O$ and $T_{i}$, i.e.\ $P=\rho_{i}T_{i}$, equivalently $\rho_{j}=\rho_{i}\rho_{ij}$ for all $j$, then $\rhoM=|\rho_{i}|$ and $P'=\operatorname{sgn}(\rho_{i})\,T_{i}$ (if $P\ne O$): the other variables do not improve linear prediction of $Y$.
(iii) The inequality $\rhoM^2>\sum\rho_i^2$ studied by Hamilton \citep{Hamilton1987} and known as ``enhancement'' \citep{Waller2011}, is the situation in which $E$ is narrower than the unit sphere in the direction of $P$: by Theorem~\ref{thm:multiple}, $\rhoM^2>|OP|^2$ iff $|OP'|<1$ (for $P\ne O$).
For $p=2$ with $\rho_{12}=0.8$ and $\rv=(0.7,0.2)\tr$, \eqref{eq:multiple} gives $\rhoM^2=0.85$, although $|OP|^2=\rho_1^2+\rho_2^2=0.53$: by Theorem~\ref{thm:multiple}, $|OP'|=|OP|/\rhoM=\sqrt{0.53/0.85}\approx0.790<1$, so $E$ is narrower than the unit sphere in the direction of $P$ (this configuration is examined in Example~\ref{ex:artificial}).

\subsection{Standardized coefficients as a linear scalar field}

\begin{theorem}[standardized regression coefficients by linear scalar fields]\label{thm:coef}
Let $f_i:\R^p\to\R$ $(1\le i\le p)$ be the linear functional determined by $f_i(T_j)=\delta_{ij}$ $(1\le j\le p)$:
\lab{a} $P=\sum_{i=1}^p\bsh{i}T_i$: the coordinates of $P$ in the basis $\{T_i\}_{i=1}^p$.
\lab{b} $f_i(P)=\bsh{i}$ for $i=1,\dots,p$ (standardized regression coefficient).
\end{theorem}

\begin{proof}
\labfirst{a} $T_1,\dots,T_p$ are the columns of the invertible matrix $\Rm$, hence a basis. By \eqref{eq:normal} in (F1), $\Rm\ahs=\rv$, which reads $P=\sum_{i=1}^p\bsh{i}T_i$.

\lab{b} Since $T_1,\dots,T_p$ form a basis, each $f_i$ exists and is unique; explicitly,
\begin{equation}\label{eq:field}
f_i(\bm{x})=\bm{e}_i\tr\Rm^{-1}\bm{x}=\langle \bm{e}_i,\bm{x}\rangle_{\Rm^{-1}},
\end{equation}
because $\bm{e}_i\tr\Rm^{-1}\Rm \bm{e}_j=\delta_{ij}$. Thus $f_i(P)=(\Rm^{-1}\rv)_i=\bsh{i}$.
\end{proof}

The level sets $f_i(\bm{x})=c$ are parallel hyperplanes (Figure~\ref{fig:field}): $f_i(\bm{x})=0$ is the hyperplane through $O$ spanned by $\{T_j\}_{j\ne i}$, and $f_i(\bm{x})=\pm1$ are the parallel hyperplanes through $\pm T_i$.
The shading of $\bar E$ by $f_i$ is a ``linear scalar field'' from which $\bsh{i}$ is read at $P$.
Since $f_i(\bm{x})=\langle \bm{e}_i,\bm{x}\rangle_{\Rm^{-1}}$ by \eqref{eq:field}, the hyperplane $f_i(\bm{x})=0$ is the diametral hyperplane of $E$ conjugate to the coordinate direction $\bm{e}_i$: any chord of $E$ parallel to the $x_i$-axis has its midpoint in this hyperplane. Conversely, any point of this hyperplane lying inside the ellipsoid is the midpoint of such a chord.
By Lemma~\ref{lem:inscribed}(c), this hyperplane is spanned by $\{T_j\}_{j\ne i}$.
In particular, hold $\Rm$ and also $\rv_{(-i)}$ fixed. As $P$ moves along the chord parallel to the $x_i$-axis, the standardized coefficient of $X_i$ is zero exactly at the midpoint. This midpoint is the zero of the ruler in Theorem~\ref{thm:partial} below.

\begin{figure}[t]
\centering
\includegraphics[width=0.98\textwidth]{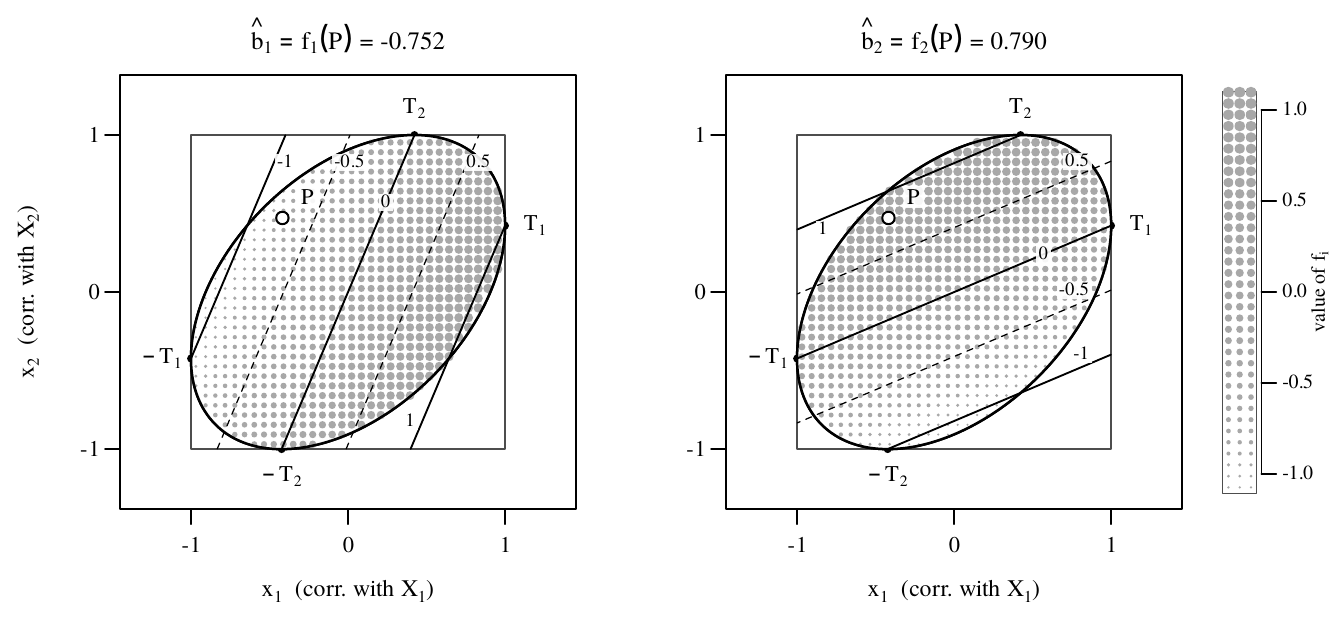}
\caption{Theorem~\ref{thm:coef} for Example~\ref{ex:baseball}. Left: the field $f_1$, with the labeled level lines $f_1(\bm{x})=-1,-0.5,0,0.5,1$; the solid lines $f_1(\bm{x})=-1,0,1$ pass through $-T_1$, $\pm T_2$ and $T_1$ respectively. Right: the field $f_2$. The field is rendered as a halftone inside $\bar E$: gray dots on a white ground vary in size, and their area---hence the covered fraction of the ground---grows linearly with the value of $f_i$, from $-\sqrt{\VIF_i}$ (almost bare) to $\sqrt{\VIF_i}$ (almost covered); see Prop.~\ref{prop:bounds}(c).}
\label{fig:field}
\end{figure}

\subsection{Partial correlation as a ruler reading on a chord}\label{sec:partial}

\begin{assumption}[chord of positive length]\label{ass:chord}
\begin{equation}\label{eq:chord}\tag{B$_i$}
\text{$\operatorname{length}(\ell_i\cap\bar E)>0$ for the chosen $i$.}
\end{equation}
\end{assumption}

\begin{theorem}[partial correlation by a ruler]\label{thm:partial}
Fix $i$ and assume \eqref{eq:chord}. Let $Q_i^{\pm}$ be the two endpoints of the chord $\ell_i\cap\bar E$, and let $g_i:\R\to\ell_i$ be the affine bijection with $g_i(\pm1)=Q_i^\pm$ (a ruler laid on the chord, graduated linearly from $-1$ at $Q_i^-$ to $+1$ at $Q_i^+$). Here, the $i$-th coordinate of $Q_i^+$ is larger than that of $Q_i^-$.

\lab{a} Put $v_i:=(\Rm^{-1})_{ii}=\bm{e}_i\tr\Rm^{-1}\bm{e}_i$ and $D_i:=(\bsh{i})^2+v_i(1-\rhoM^2)$. Then
\begin{equation}\label{eq:closed}
Q_i^\pm=P+t_i^{\pm} \bm{e}_i\quad\left(t_i^{\pm}:=\frac{-\bsh{i}\pm\sqrt{D_i}}{v_i}\right),\qquad
g_i^{-1}(P)=\frac{\bsh{i}}{\sqrt{D_i}}.
\end{equation}
\lab{b} 
\[
\frac{\bsh{i}}{\sqrt{D_i}}=\rhoP{i},\qquad\text{hence}\quad g_i^{-1}(P)=\rhoP{i}\quad\text{(partial correlation)}.
\]
\end{theorem}

\begin{proof}
\labfirst{a} Since $\bm{e}_i\tr\Rm^{-1}P=(\Rm^{-1}\rv)_i=\bsh{i}$ (\eqref{eq:normal} in (F1)) and $\nR{P}^2=\rhoM^2$ (Lemma~\ref{lem:inside}(a)),
\[
\nR{P+t\bm{e}_i}^2
=\nR{P}^2+2t\langle \bm{e}_i,P\rangle_{\Rm^{-1}}+t^2\nR{\bm{e}_i}^2
=\rhoM^2+2\bsh{i} t+v_i t^2\qquad\text{for }t\in\R.
\]
The roots of $\nR{P+t\bm{e}_i}^2=1$ are the $t_i^{\pm}$ of \eqref{eq:closed}, so the chord $\ell_i\cap\bar E$ has midpoint parameter $t_i^0:=-\bsh{i}/v_i$ and half-length $\sqrt{D_i}/v_i$.
\eqref{eq:chord} says $D_i\ne0$. Thus $g_i^{-1}(P+t\bm{e}_i)=(t-t_i^0)/(\sqrt{D_i}/v_i)$.
Putting $t=0$ establishes $g_i^{-1}(P)=\bsh{i}/\sqrt{D_i}$.

\lab{b} We identify this value as the partial correlation $\rhoP{i}$.
For this part of the proof only, replace each variable by its centered version divided by its standard deviation, retaining the notation $X_i$, $Y$. This leaves all correlations unchanged, and the regression coefficients of the transformed variables are $\bsh{i}$ (Definition~\ref{def:coefficients}).
Let $\widetilde X_i^{(-i)}$ be the residual of $X_i$ and $\widetilde Y^{(-i)}$ that of $Y$, both regressed on $X_{(-i)}$, and put $\varepsilon:=Y-\ahs\tr X$, the residual of $Y$ on $X$, which is uncorrelated with every $X_j$ and has variance $1-\rhoM^2$ by \eqref{eq:residual-variance} in (F2) (with $\sd(Y)=1$). Then
\[
\widetilde X_i^{(-i)}=\frac{\bm{e}_i\tr\Rm^{-1}X}{v_i},\qquad \widetilde Y^{(-i)}=\bsh{i}\,\widetilde X_i^{(-i)}+\varepsilon,\qquad \Cov(\widetilde X_i^{(-i)},\varepsilon)=0 .
\]
For the first identity, observe that
\[
\frac{\bm{e}_i\tr\Rm^{-1}X}{v_i}
=X_i+\sum_{j\ne i}\frac{(\Rm^{-1})_{ij}}{v_i}X_j.
\]
This is $X_i$ plus a linear combination of $X_{(-i)}$. Moreover, for $j\ne i$,
\[
\begin{aligned}
\Cov(\bm{e}_i\tr\Rm^{-1}X,X_j)
&=\bm{e}_i\tr\Rm^{-1}\Cov(X,X_j)\\
&=\bm{e}_i\tr\Rm^{-1}\Rm \bm{e}_j=\delta_{ij}=0.
\end{aligned}
\]
Thus $\bm{e}_i\tr\Rm^{-1}X/v_i$ is uncorrelated with each component of $X_{(-i)}$. By the orthogonality characterization and uniqueness of the least-squares residual, it equals $\widetilde X_i^{(-i)}$.

Taking residuals on $X_{(-i)}$ in $Y=\bsh{i}X_i+\sum_{j\ne i}\bsh{j}X_j+\varepsilon$ gives the second identity: the residual of each $X_j$ ($j\ne i$) is zero, and that of $\varepsilon$ is $\varepsilon$ itself. To compute their correlation, we obtain
\begin{equation}\label{eq:residual-moments}
\Var(\widetilde X_i^{(-i)})=\frac{1}{v_i},\qquad \Cov(\widetilde X_i^{(-i)},\widetilde Y^{(-i)})=\frac{\bsh{i}}{v_i},\qquad \Var(\widetilde Y^{(-i)})=\frac{(\bsh{i})^2}{v_i}+1-\rhoM^2=\frac{D_i}{v_i} ,
\end{equation}
the first by $\Var(\bm{e}_i\tr\Rm^{-1}X)=\bm{e}_i\tr\Rm^{-1}\Rm\Rm^{-1}\bm{e}_i=v_i$. Both variances are positive by \eqref{eq:pd} and \eqref{eq:chord}, respectively, so $\rhoP{i}$ is defined, and
\[
\rhoP{i}=\frac{\bsh{i}/v_i}{\sqrt{(1/v_i)(D_i/v_i)}}=\frac{\bsh{i}}{\sqrt{D_i}}=g_i^{-1}(P). \qedhere
\]
\end{proof}

\begin{figure}[t]
\centering
\includegraphics[width=0.98\textwidth]{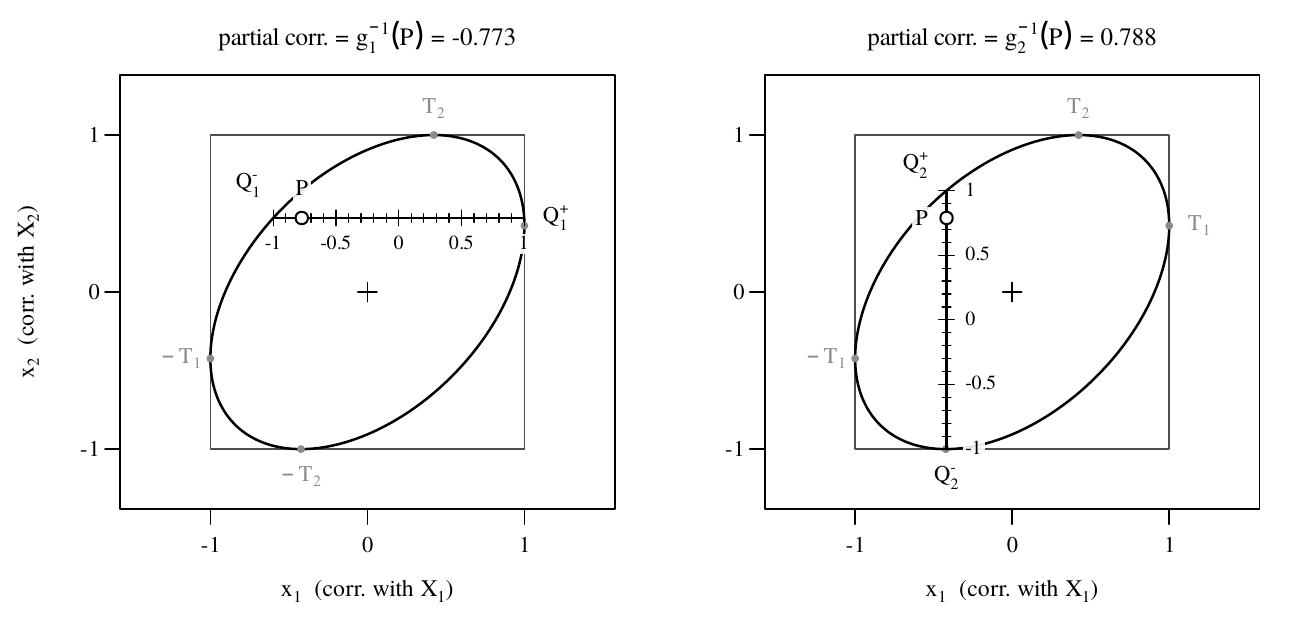}
\caption{Theorem~\ref{thm:partial} for Example~\ref{ex:baseball}. Each panel shows the chord of $E$ through $P$ parallel to the $x_1$-axis (left) or the $x_2$-axis (right); the chord carries a ruler graduated from $-1$ to $1$, and the reading at $P$ is the partial correlation. That $Q_1^+$ lies close to $T_1$ and $Q_2^-$ very close to $-T_2$ is a coincidence of this data set ($\rho_2\approx\rho_{12}$ and $\rho_1\approx-\rho_{12}$). The labels $T_1$ (left) and $-T_2$ (right) are omitted because they would overlap those of $Q_1^+$ and $Q_2^-$; $\pm T_1$ and $\pm T_2$ are drawn in gray.}
\label{fig:ruler}
\end{figure}

Figure~\ref{fig:ruler} illustrates the reading.
The squared reading has a variance interpretation: $1-\rhoM^2=(1-\rho^2_{X_{(-i)}\vdash Y})(1-\rhoP{i}^2)$ \citep[Eq.~(3.63), pp.~80--81]{TakeuchiYanai1972}.

\begin{remark}[chord $Q_i^-Q_i^+$]\label{rem:feasible}
Fix $\Rm$ and $\rv_{(-i)}$ while varying $\rho_i$. Since $\Rm\succ0$, $\Cm\succeq0$ iff $\rv\tr\Rm^{-1}\rv\le1$. Thus the admissible values of $\rho_i$ correspond to the chord $Q_i^-Q_i^+$. Its midpoint is at $\rho_i=\rv_{(-i)}\tr\Rm_{(-i)(-i)}^{-1}\Rm_{(-i)i}$, and its half-length is the denominator of \eqref{eq:partial-classical} in (F3). Under assumption \eqref{eq:chord}, Theorem~\ref{thm:partial} shows that partial correlation increases linearly from $-1$ to $+1$ along the chord. Feasibility here concerns only the correlation matrix; fixing the marginal distributions may restrict the interval further.
\end{remark}

\section{Consequences}\label{sec:consequences}

\subsection{Signs, suppression and sign reversal}

\begin{proposition}[signs and sign reversal]\label{prop:signs}
Assume \eqref{eq:chord}:
\lab{a} $\sgn\bsh{i}=\sgn\rhoP{i}=\sgn f_i(P)$ (and $\sgn\hat b_i=\sgn\bsh{i}$ by Definition~\ref{def:coefficients}), whereas $\sgn\rho_i=\sgn x_i(P)$.
\lab{b} $\rho_i\bsh{i}<0$ iff $P$ lies in $\{\bm{x} : x_if_i(\bm{x})<0\}$.
\lab{c} The hyperplanes $\{\bm{x} : x_i=0\}$ and $\{\bm{x} : f_i(\bm{x})=0\}$ coincide iff $X_i$ is uncorrelated with every component of $X_{(-i)}$, so that $f_i(\bm{x})=x_i$.
\lab{d} For $p=2$ and $i=1$, the signs disagree iff $0<\rho_1<\rho_{12}\rho_2$ or $\rho_{12}\rho_2<\rho_1<0$.
\end{proposition}
\begin{proof}
\labfirst{a} By Theorem~\ref{thm:coef}(b), $f_i(P)=\bsh{i}$; by Theorem~\ref{thm:partial}(a),(b), $\rhoP{i}=\bsh{i}/\sqrt{D_i}$ with $\sqrt{D_i}>0$ under \eqref{eq:chord}.
\lab{b} Theorem~\ref{thm:coef}(b) gives the sign-reversal criterion.
\lab{c} The zero set of $f_i$ is $\{x_i=0\}$ iff $\Rm^{-1}\bm{e}_i$ is a multiple of $\bm{e}_i$, i.e.\ iff $\Rm \bm{e}_i=c\bm{e}_i$, and then $c=\bm{e}_i\tr\Rm \bm{e}_i=1$.
\lab{d} Follows from $f_1(\bm{x})=(x_1-\rho_{12}x_2)/(1-\rho_{12}^2)$ for $p=2$.
\end{proof}

Figure~\ref{fig:wedge} shows this region of sign reversal (for $i=2$ in the configuration of Example~\ref{ex:artificial}), a pattern discussed under the heading of suppression \citep[Sects.~2--3]{FriedmanWall2005}. (When $\rho_1=0$ but $\rho_{12}\rho_2\ne0$, $\bsh{1}=-\rho_{12}\rho_2/(1-\rho_{12}^2)\ne0$ and $P$ lies on the boundary $x_1=0$: a nonzero coefficient, but not a sign reversal.)
For $p=2$ the wedges are narrow when $|\rho_{12}|$ is small.

\begin{figure}[t]
\centering
\includegraphics[width=0.46\textwidth]{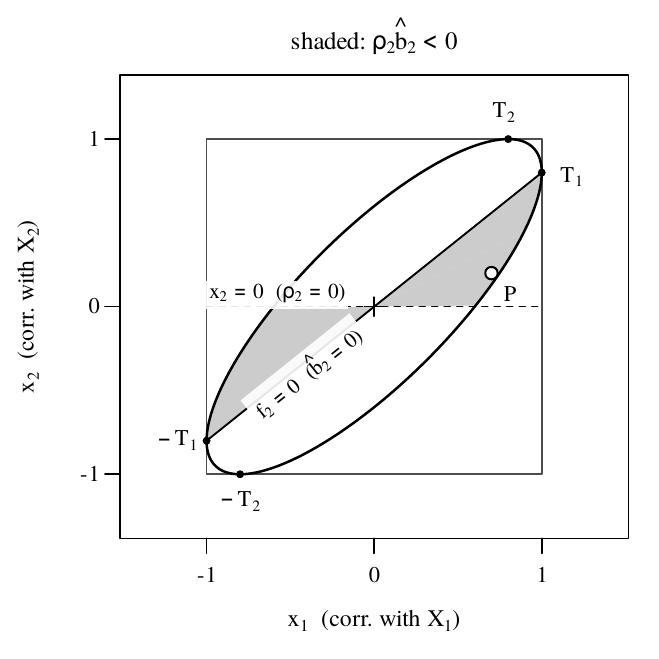}
\caption{Proposition~\ref{prop:signs}: the two opposite wedges (shaded inside $\bar E$) between $x_{2}=0$ and $f_{2}(\bm{x})=0$ in which the simple correlation $\rho_{2}$ and the standardized coefficient $\bsh{2}$ (equivalently the partial correlation) have strictly opposite signs, $\rho_{2}\bsh{2}<0$, drawn for Example~\ref{ex:artificial} ($\rho_{12}=0.8$, $\rv=(0.7,0.2)\tr$), in which $P$ lies inside one of them. Figures~\ref{fig:config} to~\ref{fig:ruler} use Example~\ref{ex:baseball}.}
\label{fig:wedge}
\end{figure}

\subsection{Bounds, variance inflation and the shape of $E$}

\begin{proposition}[VIF and coefficient bounds]\label{prop:bounds}
Put $\VIF_i:=(\Rm^{-1})_{ii}$ (cf.\ Theorem~\ref{thm:partial}):
\lab{a} The variance inflation factor of $X_i$ is
\begin{equation}\label{eq:vif}
\VIF_i=\frac{1}{1-\rho^2_{X_{(-i)}\vdash X_i}}.
\end{equation}
\lab{b} The chord of $E$ through $O$ parallel to the $x_i$-axis has
\[
\text{half-length}=1/\sqrt{\VIF_i},\qquad\text{(VIF)}
\]
whereas the orthogonal projection (shadow) of $\bar E$ onto the $x_i$-axis is $[-1,1]$.
\lab{c} $\max_{\bm{x}\in E}f_i(\bm{x})=\max_{\bm{x}\in E}|f_i(\bm{x})|=\sqrt{\VIF_i}$, and
\[
|\bsh{i}|\le\rhoM\sqrt{\VIF_i}\le\sqrt{\VIF_i},\qquad |\bsh{i}|\le\frac{\rhoM}{\sqrt{1-\rho_{12}^2}}\ \ (\text{if } p=2).
\]
\end{proposition}
\noindent\eqref{eq:vif} is given in \citet[Cor.~5.8.1]{Whittaker1990} and in \citet[p.~85, Exercise~(5)]{TakeuchiYanai1972}.
\begin{proof}
\labfirst{a} Let $\widetilde X_i^{(-i)}$ be the residual of $X_i$ regressed on $X_{(-i)}$ (with an intercept). The first equality in \eqref{eq:residual-moments} in the proof of Theorem~\ref{thm:partial}, applied to the standardized explanatory variables, gives $\Var(\widetilde X_i^{(-i)})=\sd(X_i)^2/\VIF_i$ after restoring the scale; this identity does not require the chord assumption \eqref{eq:chord}. By \eqref{eq:residual-variance} in (F2), applied to the regression of $X_i$ on $X_{(-i)}$, the same variance is $\sd(X_i)^2(1-\rho^2_{X_{(-i)}\vdash X_i})$. Dividing by $\sd(X_i)^2>0$ proves \eqref{eq:vif}. For $p=1$, both sides of \eqref{eq:vif} equal $1$ with the conventions of Section~\ref{sec:notation}.

\lab{b} Since $\nR{t\bm{e}_i}^2=t^2\bm{e}_i\tr\Rm^{-1}\bm{e}_i=\VIF_i\,t^2$ for every $t\in\R$, the point $t\bm{e}_i$ lies on $E$ iff $t=\pm1/\sqrt{\VIF_i}$. The projection is $[-1,1]$ by the convexity of $\bar E$ and Lemma~\ref{lem:inscribed}(a).
\lab{c} By Cauchy--Schwarz and Theorem~\ref{thm:coef}, for $\bm{x}\in E$,
\[
|f_i(\bm{x})|=|\langle \bm{e}_i,\bm{x}\rangle_{\Rm^{-1}}|
\le\|\bm{e}_i\|_{\Rm^{-1}}\|\bm{x}\|_{\Rm^{-1}}
=\|\bm{e}_i\|_{\Rm^{-1}}
=\sqrt{\VIF_i}.
\]
At $\bm{x}=\bm{e}_i/\sqrt{\VIF_i}\in E$, we have $f_i(\bm{x})=\sqrt{\VIF_i}$, proving both maxima.
If $P\ne O$, write $P=\rhoM P'$ with $P'\in E$ (Theorem~\ref{thm:multiple}); then, by linearity, $|\bsh{i}|=|f_i(P)|=\rhoM|f_i(P')|\le\rhoM\sqrt{\VIF_i}$. For $P=O$, there is nothing to prove.
\end{proof}

A small eigenvalue of $\Rm$ makes $E$ thin in the corresponding eigenvector direction while $E$ still touches all $2p$ faces. For an explanatory variable with large $\VIF_i$, the chord through the center parallel to the $x_i$-axis is short, and the level hyperplanes of $f_i$ are closely spaced. When $P$ points along a particularly thin direction of $E$, it can lie close to $E$ even though all its coordinates are small, giving a multiple correlation close to one, as in the example after Theorem~\ref{thm:multiple}.

\section{Worked examples}\label{sec:example}

\begin{example}[league standings: rank, runs scored and runs allowed]\label{ex:baseball}
Let $Y$ be a team's final rank ($1=$ champion), $X_1$ its season total of runs scored, and $X_2$ its season total of runs allowed in Japan's six-team Central League.
The data were transcribed by the author from \emph{The Official Baseball Encyclopedia 2004} (in Japanese) \citep{NPBEncyclopedia2004}.
From 1954--2003, we use the 31 seasons in which every team played 130 games, so that season totals are comparable, giving 186 team--season observations.
The observations are listed in Appendix~\ref{app:data}; the data and R script are supplied as ancillary files.
The pooled Pearson correlations, rounded to three decimals, are $\rho_1=-0.418$, $\rho_2=0.472$, and $\rho_{12}=0.423$.
Figures~\ref{fig:config} to~\ref{fig:ruler} and the following table use these rounded values.
\begin{center}
\begin{tabular}{lcc}\toprule
quantity & geometric reading & algebra (F1)--(F3)\\ \midrule
$\rhoM$ (Thm.~\ref{thm:multiple}) & $|OP|/|OP'|$ & $0.829$\\
$\bsh{1},\bsh{2}$ (Thm.~\ref{thm:coef}) & $f_1(P),f_2(P)$ & $-0.752,\ 0.790$\\
$\rhoP1,\rhoP2$ (Thm.~\ref{thm:partial}) & $g_1^{-1}(P),g_2^{-1}(P)$ & $-0.773,\ 0.788$\\
$1/\sqrt{\VIF_i}$, $i=1,2$ (Prop.~\ref{prop:bounds}) & half-length of the chord through $O$ & $0.906$\\ \bottomrule
\end{tabular}
\end{center}
The diagram shows why the multiple correlation ($0.83$) is so much larger than either simple correlation ($0.42$, $0.47$): $P$ points nearly along the short axis of $E$, the direction ``runs scored minus runs allowed'', i.e.\ the run differential.
\end{example}

\begin{example}[an artificial configuration with sign reversal and a coefficient exceeding one]\label{ex:artificial}
Take $\rho_{12}=0.8$ and $\rv=(0.7,0.2)\tr$, the configuration of Section~\ref{sec:theorems} (the full correlation matrix is positive definite, with determinant $0.054$).
By \eqref{eq:normal} in (F1) and \eqref{eq:multiple} in (F2),
\[
\ahs=\frac{1}{1-0.8^2}
\begin{pmatrix}0.7-0.8\times0.2\\0.2-0.8\times0.7\end{pmatrix}
=\begin{pmatrix}1.5\\-1\end{pmatrix},
\qquad
\rhoM^2=0.7\times1.5+0.2\times(-1)=0.85.
\]
Then $\rhoM=\sqrt{0.85}\approx0.922$ and $(\rhoP1,\rhoP2)=(\sqrt{27/32},-\sqrt{12/17})\approx(0.919,-0.840)$: the standardized coefficient of $X_1$ exceeds $1$, and the standardized coefficient and partial correlation of $X_2$ have the opposite sign to $\rho_2=0.2$.
$E$ is thin (semi-axes $\sqrt{1.8}$ and $\sqrt{0.2}$).
$P$ lies outside the strip $|f_1(\bm{x})|\le1$ between the level lines through $\pm T_1$ (Theorem~\ref{thm:coef}), and inside the sign-reversal wedge for $i=2$ shaded in Figure~\ref{fig:wedge}, since $\rho_2=0.2$ is strictly between $0$ and $\rho_{12}\rho_1=0.56$ (Proposition~\ref{prop:signs}).
Here $\VIF_1=\VIF_2=1/(1-0.8^2)=25/9$: the chords of $E$ through $O$ parallel to the axes have half-length $1/\sqrt{\VIF_i}$$=0.6$ although $E$ still touches all four sides of the square. Moreover, the bound $|\bsh{i}|\le\rhoM\sqrt{\VIF_i}=\sqrt{85}/6\approx1.537$ of Proposition~\ref{prop:bounds}(c) is nearly attained by $\bsh{1}=1.5$.
\end{example}

The readings in Theorems~\ref{thm:multiple}, \ref{thm:coef} and~\ref{thm:partial} and the central-chord/VIF identity in Proposition~\ref{prop:bounds}(b) and~(a) were also checked numerically against (F1)--(F3) on 1\,000 random positive definite correlation matrices for $p\in\{2,3,4,6,10\}$; a shortened script for Theorems~\ref{thm:multiple}, \ref{thm:coef} and~\ref{thm:partial} is reproduced in Appendix~\ref{app:code}, and the full script is provided as an ancillary file.

\section{Relation to known elliptical constructions}\label{sec:related}

We now compare the configuration with the elliptical and vector-geometric constructions already in use. The ellipsoid $E$ is the classical data ellipsoid of the standardized explanatory variables. Furthermore, enhancement and suppression, feasible correlations, and the vector geometry of regression are related to the configuration.

\paragraph{The ellipsoid $E$ is the data ellipsoid of the standardized explanatory variables.}
For standardized variables the ellipsoid of concentration of Dempster \citep[Sect.~7.3, Eqs.~(7.3.2)--(7.3.3)]{Dempster1969} is $\{\bm{x} : \nR{\bm{x}}\le1\}=\bar E$; Cram\'er's original definition \citep[Sect.~22.7, Eq.~(22.7.1)]{Cramer1946} scales it by $\sqrt{p+2}$, so that the uniform distribution on the solid has the given second moments.
Its shadows on the coordinate axes are $[-1,1]$ by Dempster's shadow theorem \citep[Thm.~7.3]{Dempster1969}, in the words of \citet[Fig.~2 and Sect.~3]{FriendlyMonetteFox2013}, its ``shadows (projections) on the coordinate axes give standard intervals'' $\bar x\pm s_x$, $\bar y\pm s_y$, which for standardized variables is Lemma~\ref{lem:inscribed}(a).
For $p=2$ the same ellipse is, up to translation and uniform scaling, the correlation glyph of Murdoch and Chow \citep{MurdochChow1996}.
Lemma~\ref{lem:inscribed} is therefore a statement of known facts in a form (uniqueness, tangency points as columns of $\Rm$, the coordinate hyperplanes as diametral hyperplanes) convenient for what follows.
It differs from the ``conjugate containment'' of Dempster \citep[Sect.~3.2, pp.~38--39]{Dempster1969}, in which an ellipsoid touches each face of a parallelotope at the center of the face, and the points of contact form a set of conjugate semi-axes: here the cube is the axis-parallel bounding box of $E$, and its points of contact $T_1,\dots,T_p$ are conjugate to the coordinate directions (Lemma~\ref{lem:inscribed}(c)) but in general not to one another.

\paragraph{What is new is the point $P$ and the three readings.}
The data ellipsoid lies in the space of the predictors and is used to display the predictors' own second moments.
Here the vector of correlations with the response is placed \emph{as a point} in that same space, and Theorems~\ref{thm:multiple}, \ref{thm:coef} and~\ref{thm:partial} show that the relative position of this point, the ellipsoid and the cube encodes the three quantities considered here.
Algebraically, Theorem~\ref{thm:multiple} says that $\rhoM$ is the $\Rm^{-1}$-norm of $\rv$ (formula \eqref{eq:multiple} in (F2) read as a Mahalanobis length).
Theorem~\ref{thm:coef} says that the normal equations $\Rm\ahs=\rv$ express $\rv$ in the basis of columns of $\Rm$.
Theorem~\ref{thm:partial} gives the identity $\rhoP{i}=\bsh{i}/\sqrt{(\bsh{i})^2+(\Rm^{-1})_{ii}(1-\rhoM^2)}$, which is equivalent to \eqref{eq:partial-classical} in (F3) but, in this form, describes the position of $P$ along a chord (Remark~\ref{rem:feasible}).
We have not found the placement of $\rv$ as a point in the data ellipsoid of the explanatory variables, nor the three readings from that one figure, in the elliptical methods of \citep{Dempster1969,FriendlyMonetteFox2013}, in the vector-geometric treatments \citep{SavilleWood1991,Wickens1995,ONeill2019}, in the coefficient-space construction of Waller \citep{Waller2011}, or elsewhere in the literature we have examined.

\paragraph{Enhancement and suppression.}
For fixed $\Rm$ and fixed $R^2$, Waller \citep[Sects.~2--3]{Waller2011} describes enhancement, $R^2>\rv\tr\rv$, in coefficient space, where its boundary is the intersection of two ellipsoids.
The linear map $\bm{b}\mapsto\Rm\bm{b}$ carries $\{\bm{b} : \bm{b}\tr\Rm\bm{b}=c\}$ to $\sqrt c\,E$ and $\{\bm{b} : \bm{b}\tr\Rm^2\bm{b}=c\}$ to the sphere of radius $\sqrt c$; with $c=R^2$ these are $\rhoM E$, the homothetic copy of $E$ through $P$ (Section~\ref{sec:theorems}), and the sphere of radius $\rhoM$, so enhancement means $|OP|<\rhoM$, i.e.\ $|OP'|<1$.
Friedman and Wall \citep[Sect.~3]{FriedmanWall2005} plot $R^2$ and the coefficients against $\rho_{12}$ for fixed $\rho_1,\rho_2$; Proposition~\ref{prop:signs} shows the same phenomenon in the fixed-$\rho_{12}$ diagram, as a wedge in the position of $P$.

\paragraph{Feasible correlations.} Remark~\ref{rem:feasible} reads the admissible range of $\rho_i$ as the chord $Q_i^-Q_i^+$. Kiesl and R\"assler \citep[Sect.~4]{KieslRassler2006} describe feasible correlation vectors using an ellipsoid when some variables are never jointly observed.

\paragraph{The vector geometry of regression.}
The vector geometry in $\R^n$ \citep[Chap.~3]{TakeuchiYanai1972}, \citep{SavilleWood1991,Wickens1995,ONeill2019}, and in abstract form Dempster's variable space \citep[Sect.~8.3]{Dempster1969}, represents each variable by a vector and correlations by cosines (the multiple and partial correlations are the cosines in Eqs.~(3.41) and~(3.50) of Takeuchi and Yanai); it explains the same formulas by projections in the space spanned by the $p+1$ variable vectors, and encodes correlations as angles between vectors rather than, as here, as coordinates of points in a fixed cube.

\section{Conclusion}\label{sec:discussion}
A single ellipsoid and a point encode the multiple correlation, standardized regression coefficients and partial correlations. This geometric representation also makes sign reversal, coefficient bounds and variance inflation visible.

\section*{Acknowledgments}
The author is grateful to the library of the Institute of Statistical Mathematics, whose rich collection made the references used here readily available.

\appendix
\setcounter{table}{0}\renewcommand{\thetable}{\Alph{section}.\arabic{table}}

\section{Sources of the classical formulas}\label{app:sources}
Table~\ref{tab:sources} summarizes sources for the classical relations used in this paper, ordered by the publication year of the cited edition.
$\circ$: an equivalent relation, possibly after standardization or rearrangement;
$\triangle$: a geometric interpretation;
---: no specific passage is listed.
\begin{table}[ht]
\centering\small
\caption{Sources of the classical relations of Section~\ref{sec:notation}.}\label{tab:sources}
\setlength{\tabcolsep}{4pt}
\begin{tabular}{@{}>{\raggedright\arraybackslash}p{5.3cm}lccccc@{}}\toprule
Formula & Eq. & \shortstack[c]{Takeuchi and\\Yanai (1972)} & \shortstack[c]{Rao\\(1973)} & \shortstack[c]{\makebox[1.3cm][c]{Whittaker}\\(1990)} & \shortstack[c]{Anderson\\(2003)} & \shortstack[c]{Friendly\\et al.\ (2013)}\\ \midrule
Standardized regression coefficients & \eqref{eq:normal} & $\circ$ & $\circ$ & --- & $\circ$ & $\circ$\\
Multiple correlation & \eqref{eq:multiple} & $\circ$ & $\circ$ & --- & $\circ$ & ---\\
Residual variance & \eqref{eq:residual-variance} & $\circ$ & $\circ$ & --- & $\circ$ & ---\\
Determinant identity & \eqref{eq:det} & $\circ$ & --- & $\circ$ & $\circ$ & ---\\
Partial corr., residual form & \eqref{eq:partial-classical} & $\circ$ & --- & --- & $\circ$ & $\triangle$\\
Partial corr., inverse-matrix form & \eqref{eq:partial-inverse} & $\circ$ & --- & $\circ$ & --- & ---\\
Variance inflation factor & \eqref{eq:vif} & $\circ$ & --- & $\circ$ & --- & $\triangle$\\ \bottomrule
\end{tabular}
\end{table}

\noindent The corresponding passages are:
\begin{itemize}
\item \citet{TakeuchiYanai1972}: \eqref{eq:normal}: (3.23); \eqref{eq:multiple}: the display after (3.42), p.~73; \eqref{eq:residual-variance}: (3.45), p.~74; \eqref{eq:det}: (3.44), rearranged; \eqref{eq:partial-classical}: (3.51), with $z=X_i$; \eqref{eq:partial-inverse}: (3.56), in cofactor form; \eqref{eq:vif}: Exercise~(5), p.~85.
\item \citet{Rao1973}: \eqref{eq:normal}: (4g.1.8); \eqref{eq:multiple}: (4g.1.11); \eqref{eq:residual-variance}: (4g.1.10)--(4g.1.11), p.~266; the residual definition of partial correlation: (4g.2.9), p.~270.
\item \citet{Whittaker1990}: \eqref{eq:det}: Prop.~5.9.1; \eqref{eq:partial-inverse}: Cor.~5.8.2; \eqref{eq:vif}: Cor.~5.8.1.
\item \citet{Anderson2003}: \eqref{eq:normal}: Def.~2.5.1 and the remark after Thm.~2.5.3; \eqref{eq:multiple}: Def.~2.5.4 and Eq.~(15); \eqref{eq:residual-variance}: Eq.~(19), p.~38; \eqref{eq:det}: Eqs.~(16)--(19); \eqref{eq:partial-classical}: Def.~2.5.2 and Eq.~(18).
\item Friendly, Monette and Fox \citep{FriendlyMonetteFox2013}: \eqref{eq:normal}: Eq.~(19), for centered predictors scaled to unit length, gives \eqref{eq:normal} when the centered response is also scaled to unit length; partial correlation: item~(3) preceding Fig.~18; \eqref{eq:vif}: the horizontal-width ratio in Fig.~18. Their width ratio $\sqrt{\VIF}$ and our central-chord half-length $1/\sqrt{\VIF}$ describe different constructions.
\end{itemize}
\noindent Anderson (2003) and Rao (1973) state these relations for general covariance (dispersion) matrices; the forms in $\Rm$ and $\rv$ used in Section~\ref{sec:notation} are the forms obtained by centering each variable and dividing it by its standard deviation (Definition~\ref{def:coefficients}).

\section{List of symbols}\label{app:symbols}
The tables below collect the notation of the paper. ``Where'' gives the section, theorem or equation in which the symbol is introduced or mainly used.

\begingroup\footnotesize
\renewcommand{\arraystretch}{1.3}
\begin{longtable}{@{}>{\raggedright\arraybackslash}p{2.85cm}>{\raggedright\arraybackslash}p{4.6cm}>{\raggedright\arraybackslash}p{5.15cm}>{\raggedright\arraybackslash}p{2.13cm}@{}}
\caption{Variables, correlations and classical regression quantities.}\label{tab:sym1}\\
\toprule
Symbol & Meaning & Definition / formula & Where\\ \midrule
\endfirsthead
\multicolumn{4}{@{}l}{\footnotesize\emph{(continued)}}\\[2pt]
\toprule
Symbol & Meaning & Definition / formula & Where\\ \midrule
\endhead
\bottomrule
\endlastfoot
$p$ & number of explanatory variables & $p\in\mathbb N,\ 1\le p<\infty$ & \S\ref{sec:intro}\\
$X=(X_1,\dots,X_p)\tr$ & explanatory variables & finite positive variances; see Def.~\ref{def:coefficients} & \S\ref{sec:notation}\\
$Y$ & response variable & finite positive variance; see Def.~\ref{def:coefficients} & \S\ref{sec:notation}\\
$n$ & sample size & rows of the data matrix & \S\ref{sec:notation}\\
$X_{(-i)}$ & explanatory variables except $X_i$ & $(X_1,\ldots,X_{i-1},X_{i+1},\ldots,X_p)\tr$ & \S\ref{sec:notation}\\
\midrule[0.25pt]
$\rho_{ij}$, $\rho_i$ & simple correlations & $\corr(X_i,X_j)$, $\corr(X_i,Y)$ & \S\ref{sec:notation}\\
$\Rm$ & correlation matrix of $X$ & $(\rho_{ij})\in\R^{p\times p}$ & \S\ref{sec:notation}\\
\eqref{eq:pd} & positive definiteness (Assumption~\ref{ass:pd}) & $\Rm\succ0$ & \S\ref{sec:notation}\\
$\rv$ & correlations with the response & $(\rho_1,\dots,\rho_p)\tr\in\R^p$ & \S\ref{sec:notation}\\
$\Cm$ & full correlation matrix of $(X,Y)$ & $\bigl(\begin{smallmatrix}\Rm&\rv\\ \rv\tr&1\end{smallmatrix}\bigr)$ & \S\ref{sec:notation}\\
\midrule[0.25pt]
$\Rm_{(-i)(-i)}$, $\Rm_{(-i)i}$, $\rv_{(-i)}$ & submatrices & $\Rm$ without row and column $i$; column $i$ of $\Rm$ without entry $i$; $\rv$ without entry $i$ & Def.~\ref{def:corr}; (F3); Rem.~\ref{rem:feasible}\\
$\ah$, $\hat b_i$ & regression coefficients on the original scales & $\ah=(\hat b_1,\dots,\hat b_p)\tr$; $\hat Y=\hat a+\ah\tr X$ & \S\ref{sec:intro}; Def.~\ref{def:coefficients}\\
$\ahs$, $\bsh{i}$ & standardized regression coeff. & $\bsh{i}=\sd(X_i)\hat b_i/\sd(Y)$; $\ahs=\Rm^{-1}\rv$ & Def.~\ref{def:coefficients}; \eqref{eq:normal} in (F1)\\
$\hat a$ & intercept in linear regression & $\hat Y$ when $X_1=\cdots=X_p=0$ & \S\ref{sec:intro}\\
$\hat Y$ & fitted value & $\hat a+\sum_i\hat b_iX_i$ & \S\ref{sec:intro}\\
\midrule[0.25pt]
$\rho_{A\vdash B}$ & \multicolumn{2}{l}{multiple correlation of the regression of $B$ on the set $A$} & \S\ref{sec:notation}\\
$\rhoM$ & multiple correlation of $Y$ on $X$ & $\sqrt{\rv\tr\Rm^{-1}\rv}\;(=\sqrt{\ahs\tr\rv}$) & \eqref{eq:multiple}; Thm.~\ref{thm:multiple}; \eqref{eq:closed}\\
$\rhoP{i}$ & \multicolumn{2}{p{\dimexpr9.75cm+2\tabcolsep\relax}}{partial correlation of $X_i$ and $Y$ given $X_{(-i)}$, defined as the correlation of the residuals of $X_i$ and $Y$ after linear regression on $X_{(-i)}$, when both residual variances are positive} & \S\ref{sec:notation}; (F3); Thm.~\ref{thm:partial}\\
$\VIF_i$ & variance inflation factor of $X_i$ & $(\Rm^{-1})_{ii}=v_i=1/(1-\rho^2_{X_{(-i)}\vdash X_i})$ & Prop.~\ref{prop:bounds}\\
$\corr$, $\sd$, $\sgn$, $\det$ & \multicolumn{2}{p{\dimexpr9.75cm+2\tabcolsep\relax}}{correlation, standard deviation (population or sample), sign, determinant} & \S\ref{sec:notation}; Def.~\ref{def:coefficients}\\
\end{longtable}

\begin{longtable}{@{}>{\raggedright\arraybackslash}p{2.85cm}>{\raggedright\arraybackslash}p{4.6cm}>{\raggedright\arraybackslash}p{5.15cm}>{\raggedright\arraybackslash}p{2.13cm}@{}}
\caption{Geometry in $\R^p$.}\label{tab:sym2}\\
\toprule
Symbol & Meaning & Definition / formula & Where\\ \midrule
\endfirsthead
\multicolumn{4}{@{}l}{\footnotesize\emph{(continued)}}\\[2pt]
\toprule
Symbol & Meaning & Definition / formula & Where\\ \midrule
\endhead
\bottomrule
\endlastfoot
$\bm{x}=(x_1,\dots,x_p)\tr$ & coordinates of $\R^p$ & axis $i$ carries the correlation with $X_i$: $x_i(P)=\rho_i$, $x_i(T_j)=\rho_{ij}$ & Def.~\ref{def:objects} in \S\ref{sec:geometry}\\
$O$ & origin of $\R^p$ & $(0,\ldots,0)\in\R^p$ & Def.~\ref{def:objects} in \S\ref{sec:geometry}\\
$\bm{e}_i$ & standard basis vector & $1$ in position $i$, $0$ elsewhere & \S\ref{sec:notation}\\
$\delta_{ij}$ & Kronecker delta & $1$ if $i=j$; $0$ otherwise & \S\ref{sec:notation}\\
$|\cdot|$ & Euclidean length; absolute value & $|AB|=\sqrt{\sum_{i=1}^{p}(A_i-B_i)^2}$ & \S\ref{sec:notation}\\
\midrule[0.25pt]
$\bm H$ & symmetric positive definite matrix & $\{\bm{x} :\bm{x}\tr\bm H\bm{x}=1\}$; $\bm H=\Rm^{-1}$ for $E$ & \S\ref{sec:notation}; Lem.~\ref{lem:inscribed}\\
$\langle \bm{u},\bm{v}\rangle_{\bm H}$, $\|\bm{u}\|_{\bm H}$ & inner product and norm & $\bm{u}\tr\bm H\bm{v}$, $\langle \bm{u},\bm{u}\rangle_{\bm H}^{1/2}$ & \S\ref{sec:notation}\\
conjugate directions & $\langle \bm{u},\bm{v}\rangle_{\bm H}=0$ & & \S\ref{sec:notation}; Lem.~\ref{lem:inscribed}\\
diametral hyperplane & \multicolumn{2}{p{\dimexpr9.75cm+2\tabcolsep\relax}}{$\{\bm{x} : \langle \bm{d},\bm{x}\rangle_{\bm H}=0\}$ bisects the chords parallel to $\bm{d}$} & \S\ref{sec:notation}; Lem.~\ref{lem:inscribed}\\
$U$ & unit hypercube & $[-1,1]^p$ & Def.~\ref{def:objects} in \S\ref{sec:geometry}\\
$\partial U$ & boundary of $U$ (the $2p$ faces $x_i=\pm1$) & $\{\bm{x}\in U : \max_i|x_i|=1\}$ & Lem.~\ref{lem:inscribed}\\
\midrule[0.25pt]
$E$ & inscribed ellipsoid (surface) & $\{\bm{x} : \bm{x}\tr\Rm^{-1}\bm{x}=1\}$ & Def.~\ref{def:objects} in \S\ref{sec:geometry}\\
$\bar E$ & solid ellipsoid bounded by $E$ & $\{\bm{x} : \bm{x}\tr\Rm^{-1}\bm{x}\le1\}$ & Def.~\ref{def:objects} in \S\ref{sec:geometry}\\
$E_{\bm M}$ & ellipsoid defined by a positive definite correlation matrix $\bm M$ & $\{\bm{x} : \bm{x}\tr\bm M^{-1}\bm{x}=1\}$ & Lem.~\ref{lem:inscribed} in \S\ref{sec:geometry}\\
$\pm T_i$ & tangency points of $E$ with $x_i=\pm1$ & $\pm\Rm \bm{e}_i$ ($\Rm \bm{e}_i$: $i$-th column of $\Rm$) & Def.~\ref{def:objects}; Lem.~\ref{lem:inscribed}\\
$P$ & the point whose position is read & $\rv=(\rho_1,\dots,\rho_p)\tr$ & Def.~\ref{def:objects}; Lem.~\ref{lem:inside}\\
\midrule[0.25pt]
$P'$ & point where the ray $OP$ meets $E$ & $P/\rhoM$ (when $P\ne O$) & Thm.~\ref{thm:multiple} in \S\ref{sec:theorems}\\
$\rhoM E$ & homothetic copy of $E$ through $P$ & $\{\bm{x} : \nR{\bm{x}}=\rhoM\}$ & \S\ref{sec:theorems}; \S\ref{sec:related}\\
$f_i:\R^p\to\R$ (linear) & linear scalar field giving $\bsh{i}$ at $P$ & $f_i(T_j)=\delta_{ij}$; $f_i(\bm{x})=\bm{e}_i\tr\Rm^{-1}\bm{x}=\langle \bm{e}_i,\bm{x}\rangle_{\Rm^{-1}}$ & Thm.~\ref{thm:coef} in \S\ref{sec:theorems}\\
$c$ & real scalar (level value of $f_i$; also a generic scalar in proofs) & $f_i(\bm{x})=c$; $c\in\R$ & Lem.~\ref{lem:inscribed}; after Thm.~\ref{thm:coef}; \S\ref{sec:related}\\
$\ell_i\parallel x_i\text{-axis}$ & a line through $P$ & $\{P+t\bm{e}_i : t\in\R\}\parallel x_i\text{-axis}$ & Def.~\ref{def:objects} in \S\ref{sec:geometry}\\
\midrule[0.25pt]
\eqref{eq:chord} & chord of positive length (Assumption~\ref{ass:chord}) & $\operatorname{length}(\ell_i\cap\bar E)>0$ for the chosen $i$ & \S\ref{sec:partial}; Thm.~\ref{thm:partial}; Prop.~\ref{prop:signs}\\
$Q_i^-$, $Q_i^+$ & endpoints of the chord $\ell_i\cap\bar E$ & $P+t_i^{\pm} \bm{e}_i$, eq.~\eqref{eq:closed} & Thm.~\ref{thm:partial} in \S\ref{sec:partial}\\
$g_i:\R\to\ell_i$ (affine) & ruler on the chord $\ell_i\cap\bar E$ & $g_i(\pm1)=Q_i^\pm$; reading $g_i^{-1}(P)$ & Thm.~\ref{thm:partial} in \S\ref{sec:partial}\\
$v_i$ & $i$-th diagonal entry of $\Rm^{-1}$ & $(\Rm^{-1})_{ii}=\bm{e}_i\tr\Rm^{-1}\bm{e}_i=\VIF_i$ & Thm.~\ref{thm:partial} in \S\ref{sec:partial}\\
$D_i$ & reduced discriminant of the chord equation & $(\bsh{i})^2+v_i(1-\rhoM^2)$ & Thm.~\ref{thm:partial}; eq.~\eqref{eq:closed}\\
\midrule[0.25pt]
$t_i^{\pm}$ & parameters of the chord endpoints & roots of $\nR{P+t\bm{e}_i}^2=1$, eq.~\eqref{eq:closed} & Thm.~\ref{thm:partial} in \S\ref{sec:partial}\\
$\bm{b}$ & coeff. vector in Waller's construction & $\bm{b}\mapsto\Rm\bm{b}$ carries $\{\bm{b}\colon \bm{b}\tr\Rm\bm{b}=c\}$ to $\sqrt c\,E$ & \S\ref{sec:related}\\
\end{longtable}

\endgroup

\section{Numerical verification}\label{app:code}
The following R script (base R only) draws random data, computes the classical quantities (F1)--(F3), constructs $E$, $P$, $f_i$ and the chords, and asserts equality of the geometric readings of Theorems~\ref{thm:multiple}, \ref{thm:coef} and~\ref{thm:partial} with the classical values to $10^{-8}$; it passes for $p\in\{2,3,4,6,10\}$ with 200 random problems each. It is a shortened version of the file \texttt{verify\_theorems.R}, which also checks Lemma~\ref{lem:inscribed} and the central-chord/VIF identity of Proposition~\ref{prop:bounds}(b) and~(a).
{\small\begin{verbatim}
set.seed(20260921)
geometric <- function(R, P) {
  Rinv <- solve(R); p <- length(P); q <- drop(t(P) %*% Rinv %*% P)
  rho  <- sqrt(sum(P^2)) / sqrt(sum((P / sqrt(q))^2))  # Thm 9: |OP|/|OP'|
  beta <- drop(Rinv %*% P)                # Thm 10: f_i(x) = e_i' R^-1 x
  pc   <- numeric(p)                      # Thm 12: ruler on the chord
  for (i in 1:p) {
    e <- diag(p)[, i]; a <- Rinv[i, i]; b <- drop(Rinv[i, ] %*% P)
    d <- sqrt(b^2 - a * (q - 1))
    Qm <- P + (-b - d) / a * e; Qp <- P + (-b + d) / a * e
    mid <- (Qm + Qp) / 2; half <- (Qp - Qm) / 2
    pc[i] <- sum((P - mid) * half) / sum(half^2)
  }
  list(rho = rho, beta = beta, pc = pc)
}
for (p in c(2, 3, 4, 6, 10)) for (trial in 1:200) {
  A <- matrix(rnorm((p + 1)^2), p + 1); S <- A %*% t(A) + 0.3 * diag(p + 1)
  Z <- scale(matrix(rnorm(400 * (p + 1)), 400) %*% chol(S))  # 400 rows
  RXY <- cor(Z); R <- RXY[1:p, 1:p]; P <- RXY[1:p, p + 1]
  beta <- unname(coef(lm(Z[, p + 1] ~ Z[, 1:p] - 1)))   # (F1)
  rho  <- drop(cor(Z[, 1:p] %*% beta, Z[, p + 1]))         # (F2)
  K    <- solve(RXY)                                       # (F3)
  pc   <- -K[1:p, p + 1] / sqrt(diag(K)[1:p] * K[p + 1, p + 1])
  pc   <- unname(pc)
  g <- geometric(R, P)
  stopifnot(isTRUE(all.equal(g$rho,  rho,  tolerance = 1e-8)),
            isTRUE(all.equal(g$beta, beta, tolerance = 1e-8)),
            isTRUE(all.equal(g$pc,   pc,   tolerance = 1e-8)))
}
cat("all checks passed\n")
\end{verbatim}}

\section{Data of Example~\ref{ex:baseball}}\label{app:data}
The records of Example~\ref{ex:baseball} were transcribed by the author from \emph{Nihon Puro Yakyu Kiroku Daihyakka 2004} (The Official Baseball Encyclopedia 2004; Nippon Professional Baseball Organization, 2004): for each of the 50 seasons 1954--2003 and each of the six teams of the Central League, the final rank and the number of games from the team standings table, the runs scored from the team batting table and the runs allowed from the team pitching table of that season, matched by team name. The values were checked against photographs of the pages.
They are provided as ancillary files of this arXiv submission as the file \texttt{central\_league\_1954to2003.csv} (300 rows, UTF-8, comma-separated; columns \texttt{year}, \texttt{rank}, \texttt{team}, \texttt{team\_ja}, \texttt{games}, \texttt{runs\_scored}, \texttt{runs\_allowed}, \texttt{season\_all\_130}, the last being \texttt{TRUE} when every team played 130 games in that season), together with the script \texttt{compute\_correlations.R} (base R only). The script checks the file---300 rows, six teams in every season, and in every season the runs scored by the six teams sum to the runs allowed by them, since all games are within the league---then selects the seasons with \texttt{season\_all\_130 = TRUE} and prints the Pearson correlations of the pooled observations:
{\small\begin{verbatim}
seasons with 130 games for every team: 31
team-season observations: 186
  rho_1  = cor(runs scored,  rank) = -0.4180
  rho_2  = cor(runs allowed, rank) = +0.4724
  rho_12 = cor(runs scored,  runs allowed) = +0.4232
\end{verbatim}}
Table~\ref{tab:data} lists the 186 observations used. Teams are listed under the name printed in the source for each season (romanized), so a team that changed its name appears under more than one name; the four other teams in the table kept their names throughout.
\begingroup\scriptsize\renewcommand{\arraystretch}{1.05}
\begin{longtable}{@{}rlrrr@{\hspace{0.8em}}rlrrr@{\hspace{0.8em}}rlrrr@{}}
\caption{The 186 team--season observations of Example~\ref{ex:baseball}: season, team, final rank ($Y$), runs scored ($X_1$) and runs allowed ($X_2$), for the 31 seasons in which every team played 130 games.}\label{tab:data}\\
\toprule year & team & $Y$ & $X_1$ & $X_2$ & year & team & $Y$ & $X_1$ & $X_2$ & year & team & $Y$ & $X_1$ & $X_2$\\ \midrule
\endfirsthead
\multicolumn{15}{@{}l}{\footnotesize\emph{(continued)}}\\[2pt] \toprule year & team & $Y$ & $X_1$ & $X_2$ & year & team & $Y$ & $X_1$ & $X_2$ & year & team & $Y$ & $X_1$ & $X_2$\\ \midrule
\endhead
\bottomrule \endlastfoot
1954 & Chunichi & 1 & 510 & 376 & 1971 & Taiyo & 3 & 363 & 354 & 1981 & Chunichi & 5 & 545 & 537\\
1954 & Kyojin & 2 & 613 & 385 & 1971 & Hiroshima & 4 & 422 & 444 & 1981 & Taiyo & 6 & 444 & 609\\
1954 & Hanshin & 3 & 569 & 425 & 1971 & Hanshin & 5 & 333 & 412 & 1982 & Chunichi & 1 & 529 & 478\\
1954 & Hiroshima & 4 & 419 & 581 & 1971 & Yakult & 6 & 427 & 451 & 1982 & Kyojin & 2 & 511 & 439\\
1954 & Kokutetsu & 5 & 434 & 518 & 1972 & Kyojin & 1 & 587 & 492 & 1982 & Hanshin & 3 & 517 & 496\\
1954 & Yosho & 6 & 365 & 625 & 1972 & Hanshin & 2 & 470 & 431 & 1982 & Hiroshima & 4 & 513 & 493\\
1955 & Kyojin & 1 & 579 & 291 & 1972 & Chunichi & 3 & 452 & 491 & 1982 & Taiyo & 5 & 473 & 555\\
1955 & Chunichi & 2 & 420 & 334 & 1972 & Yakult & 4 & 530 & 543 & 1982 & Yakult & 6 & 437 & 519\\
1955 & Hanshin & 3 & 423 & 392 & 1972 & Taiyo & 5 & 481 & 523 & 1983 & Kyojin & 1 & 622 & 542\\
1955 & Hiroshima & 4 & 395 & 508 & 1972 & Hiroshima & 6 & 482 & 522 & 1983 & Hiroshima & 2 & 610 & 527\\
1955 & Kokutetsu & 5 & 424 & 428 & 1973 & Kyojin & 1 & 565 & 470 & 1983 & Taiyo & 3 & 561 & 631\\
1955 & Taiyo & 6 & 290 & 578 & 1973 & Hanshin & 2 & 433 & 404 & 1983 & Hanshin & 4 & 598 & 604\\
1956 & Kyojin & 1 & 568 & 351 & 1973 & Chunichi & 3 & 441 & 433 & 1983 & Chunichi & 5 & 563 & 603\\
1956 & Hanshin & 2 & 386 & 283 & 1973 & Yakult & 4 & 377 & 406 & 1983 & Yakult & 6 & 557 & 604\\
1956 & Chunichi & 3 & 369 & 319 & 1973 & Taiyo & 5 & 423 & 480 & 1984 & Hiroshima & 1 & 635 & 494\\
1956 & Kokutetsu & 4 & 346 & 405 & 1973 & Hiroshima & 6 & 390 & 436 & 1984 & Chunichi & 2 & 641 & 551\\
1956 & Hiroshima & 5 & 344 & 491 & 1974 & Chunichi & 1 & 561 & 537 & 1984 & Kyojin & 3 & 637 & 517\\
1956 & Taiyo & 6 & 319 & 483 & 1974 & Kyojin & 2 & 589 & 460 & 1984 & Hanshin & 4 & 563 & 637\\
1957 & Kyojin & 1 & 471 & 393 & 1974 & Yakult & 3 & 453 & 458 & 1984 & Yakult & 5 & 530 & 653\\
1957 & Hanshin & 2 & 462 & 375 & 1974 & Hanshin & 4 & 457 & 498 & 1984 & Taiyo & 6 & 484 & 638\\
1957 & Chunichi & 3 & 404 & 347 & 1974 & Taiyo & 5 & 528 & 596 & 1985 & Hanshin & 1 & 731 & 593\\
1957 & Kokutetsu & 4 & 418 & 425 & 1974 & Hiroshima & 6 & 448 & 487 & 1985 & Hiroshima & 2 & 615 & 570\\
1957 & Hiroshima & 5 & 352 & 436 & 1975 & Hiroshima & 1 & 510 & 421 & 1985 & Kyojin & 3 & 616 & 562\\
1957 & Taiyo & 6 & 355 & 486 & 1975 & Chunichi & 2 & 521 & 466 & 1985 & Taiyo & 4 & 589 & 653\\
1958 & Kyojin & 1 & 521 & 370 & 1975 & Hanshin & 3 & 477 & 487 & 1985 & Chunichi & 5 & 537 & 596\\
1958 & Hanshin & 2 & 444 & 387 & 1975 & Yakult & 4 & 437 & 470 & 1985 & Yakult & 6 & 550 & 664\\
1958 & Chunichi & 3 & 423 & 381 & 1975 & Taiyo & 5 & 484 & 548 & 1986 & Hiroshima & 1 & 511 & 420\\
1958 & Kokutetsu & 4 & 359 & 478 & 1975 & Kyojin & 6 & 473 & 510 & 1986 & Kyojin & 2 & 600 & 437\\
1958 & Hiroshima & 5 & 407 & 459 & 1976 & Kyojin & 1 & 661 & 505 & 1986 & Hanshin & 3 & 574 & 536\\
1958 & Taiyo & 6 & 357 & 436 & 1976 & Hanshin & 2 & 602 & 500 & 1986 & Taiyo & 4 & 482 & 547\\
1959 & Kyojin & 1 & 551 & 382 & 1976 & Hiroshima & 3 & 586 & 565 & 1986 & Chunichi & 5 & 432 & 519\\
1959 & Hanshin & 2 & 388 & 362 & 1976 & Chunichi & 4 & 519 & 610 & 1986 & Yakult & 6 & 480 & 620\\
1959 & Chunichi & 2 & 477 & 428 & 1976 & Yakult & 5 & 504 & 555 & 1987 & Kyojin & 1 & 608 & 447\\
1959 & Kokutetsu & 4 & 422 & 491 & 1976 & Taiyo & 6 & 482 & 619 & 1987 & Chunichi & 2 & 571 & 492\\
1959 & Hiroshima & 5 & 401 & 414 & 1977 & Kyojin & 1 & 648 & 504 & 1987 & Hiroshima & 3 & 549 & 450\\
1959 & Taiyo & 6 & 352 & 514 & 1977 & Yakult & 2 & 593 & 553 & 1987 & Yakult & 4 & 553 & 637\\
1960 & Taiyo & 1 & 411 & 361 & 1977 & Chunichi & 3 & 575 & 614 & 1987 & Taiyo & 5 & 516 & 601\\
1960 & Kyojin & 2 & 444 & 459 & 1977 & Hanshin & 4 & 558 & 590 & 1987 & Hanshin & 6 & 437 & 607\\
1960 & Hanshin & 3 & 436 & 403 & 1977 & Hiroshima & 5 & 603 & 672 & 1988 & Chunichi & 1 & 549 & 483\\
1960 & Hiroshima & 4 & 426 & 403 & 1977 & Taiyo & 6 & 643 & 687 & 1988 & Kyojin & 2 & 518 & 442\\
1960 & Chunichi & 5 & 436 & 454 & 1978 & Yakult & 1 & 648 & 633 & 1988 & Hiroshima & 3 & 447 & 442\\
1960 & Kokutetsu & 6 & 388 & 461 & 1978 & Kyojin & 2 & 606 & 508 & 1988 & Taiyo & 4 & 514 & 542\\
1961 & Kyojin & 1 & 435 & 383 & 1978 & Hiroshima & 3 & 713 & 596 & 1988 & Yakult & 5 & 496 & 534\\
1961 & Chunichi & 2 & 427 & 385 & 1978 & Taiyo & 4 & 596 & 554 & 1988 & Hanshin & 6 & 444 & 525\\
1961 & Kokutetsu & 3 & 387 & 361 & 1978 & Chunichi & 5 & 509 & 635 & 1989 & Kyojin & 1 & 520 & 358\\
1961 & Hanshin & 4 & 405 & 411 & 1978 & Hanshin & 6 & 532 & 678 & 1989 & Hiroshima & 2 & 510 & 431\\
1961 & Hiroshima & 5 & 417 & 464 & 1979 & Hiroshima & 1 & 601 & 523 & 1989 & Chunichi & 3 & 538 & 524\\
1961 & Taiyo & 6 & 418 & 485 & 1979 & Taiyo & 2 & 595 & 562 & 1989 & Yakult & 4 & 507 & 559\\
1969 & Kyojin & 1 & 599 & 472 & 1979 & Chunichi & 3 & 589 & 566 & 1989 & Hanshin & 5 & 497 & 580\\
1969 & Hanshin & 2 & 402 & 353 & 1979 & Hanshin & 4 & 559 & 602 & 1989 & Taiyo & 6 & 474 & 594\\
1969 & Taiyo & 3 & 444 & 470 & 1979 & Kyojin & 5 & 554 & 549 & 1994 & Kyojin & 1 & 516 & 483\\
1969 & Chunichi & 4 & 463 & 470 & 1979 & Yakult & 6 & 550 & 646 & 1994 & Chunichi & 2 & 535 & 481\\
1969 & Atoms & 5 & 495 & 569 & 1980 & Hiroshima & 1 & 544 & 476 & 1994 & Hiroshima & 3 & 585 & 584\\
1969 & Hiroshima & 6 & 414 & 483 & 1980 & Yakult & 2 & 562 & 462 & 1994 & Yakult & 4 & 486 & 585\\
1970 & Kyojin & 1 & 499 & 365 & 1980 & Kyojin & 3 & 507 & 450 & 1994 & Hanshin & 4 & 503 & 500\\
1970 & Hanshin & 2 & 435 & 335 & 1980 & Taiyo & 4 & 534 & 587 & 1994 & Yokohama & 6 & 543 & 535\\
1970 & Taiyo & 3 & 436 & 399 & 1980 & Hanshin & 5 & 521 & 549 & 1996 & Kyojin & 1 & 563 & 478\\
1970 & Hiroshima & 4 & 393 & 421 & 1980 & Chunichi & 6 & 495 & 639 & 1996 & Chunichi & 2 & 641 & 599\\
1970 & Chunichi & 5 & 435 & 462 & 1981 & Kyojin & 1 & 517 & 418 & 1996 & Hiroshima & 3 & 670 & 597\\
1970 & Yakult & 6 & 336 & 552 & 1981 & Hiroshima & 2 & 605 & 519 & 1996 & Yakult & 4 & 536 & 560\\
1971 & Kyojin & 1 & 538 & 434 & 1981 & Hanshin & 3 & 523 & 480 & 1996 & Yokohama & 5 & 571 & 650\\
1971 & Chunichi & 2 & 437 & 425 & 1981 & Yakult & 4 & 520 & 591 & 1996 & Hanshin & 6 & 482 & 579\\
\end{longtable}\endgroup

\end{document}